\documentclass{COMNET}

\usepackage{amsmath}
\usepackage{amssymb}
\usepackage{graphicx}
\usepackage{xcolor}
\usepackage{url} 
\usepackage{multirow}

\AtBeginDocument{
  \label{CorrectFirstPageLabel}
  
}

\usepackage{algorithm}
\usepackage{algpseudocode}
\algrenewcommand\algorithmicrequire{\textbf{Input:}}
\algrenewcommand\algorithmicensure{\textbf{Output:}}
\algnewcommand\Input{\item[\algorithmicinput]}%
\algnewcommand\Output{\item[\algorithmicoutput]}%

\newcommand{\argmin}[1]{\underset{#1}{\mathrm{argmin}}}

\newcommand{\vol}[1]{ \mathrm{vol}({#1}) }

\newcommand{\mi}{\imath}
\newcommand{\calS}{\mathcal{S}}
\newcommand{\sym}[1]{ {#1}_{\mathrm{sym}} }

\DeclareMathAlphabet{\mathcal}{OMS}{cmsy}{m}{n}

\begin{document}

\title{An iterative spectral algorithm for digraph clustering}

\shorttitle{An iterative spectral algorithm for digraph clustering} 
\shortauthorlist{J. Martin \emph{et al.}} 

\author{
\name{James Martin$^\dagger$, Tim Rogers and Luca Zanetti}
\address{Department of Mathematical Sciences, University of Bath, Bath, UK\email{$^\dagger$Corresponding author. Email: jlm80@bath.ac.uk}}
}

\maketitle

\begin{abstract}
{Graph clustering is a fundamental technique in data analysis with applications in many different fields. While there is a large body of work on clustering undirected graphs, the problem of clustering directed graphs is much less understood. The analysis is more complex in the directed graph case for two reasons: the clustering must preserve directional information in the relationships between clusters, and directed graphs have non-Hermitian adjacency matrices whose properties are less conducive to traditional spectral methods. Here we consider the problem of partitioning the vertex set of a directed graph into $k\ge 2$ clusters so that edges between different clusters tend to follow the same direction. We present an iterative algorithm based on spectral methods applied to new Hermitian representations of directed graphs. Our algorithm performs favourably against the state-of-the-art, both on synthetic and real-world data sets. Additionally, it is able to identify a ``meta-graph" of $k$ vertices that represents the higher-order relations between clusters in a directed graph. We showcase this capability on data sets pertaining food webs, biological neural networks, and the online card game Hearthstone.}
{graph clustering, directed graphs, spectral methods.}
\end{abstract}

\section{Introduction}
The goal of clustering is to partition a set of objects into subsets (clusters) so that objects in the same clusters are \emph{similar}, while objects in different clusters are somewhat dissimilar. When the data set is described by an undirected graph, clustering traditionally aims to find regions of the graph characterised by a high inner-density and that are connected by few inter-cluster edges~\cite{FortunatoSurvey}. 

Clustering in undirected graphs is very well understood, both from a theoretical and experimental perspective \cite{LOT,von2007tutorial,PSZ}. Many real-world data sets, however, are more faithfully represented by directed graphs. Unfortunately, clustering directed graphs is much less understood. Indeed, even the problem of formalising the notion of a cluster in a directed graph does not have a definitive answer.

In this work, we focus on the following clustering problem: we want to partition the vertex set of a digraph $G$ into $k\ge 2$ clusters $S_1,\dots,S_k$ of similar sizes so that, for any $i \ne j$, the majority of edges between $S_i$ and $S_j$ are oriented in the same direction, i.e., from $S_i$ to $S_j$ or vice versa. As we are interested in flow imbalance, we consider only oriented graphs, which are graphs with no more than one edge between any pair of vertices. Given such a clustering, we can define a small $k$-vertex digraph where a vertex $i$ is connected to a vertex $j$ whenever most edges between $S_i$ and $S_j$ in $G$ are oriented from the former to the latter. We call this graph the \emph{meta-graph} of $S_1,\dots,S_k$. Our aim becomes finding the meta-graph that best summarises $G$.

Notice that a subset in an undirected graph is typically considered a cluster if it is well connected inside and sparsely connected with respect to the rest of the graph. On the contrary, in the above formulation for digraphs, the assignment of a vertex to a cluster depends on the partitioning of all the other vertices; we are, therefore, trying to find higher-order patterns in digraphs. 

Consider, for example, the digraph in Fig.~\ref{Fig:toyexample}. The subsets of vertices shown in different colours do not exhibit particular discrepancies in their inner and outer densities. Therefore, typical undirected clustering techniques are not able to uncover specific patterns. We can observe, however, that edges between each pair of subsets have a clear dominant direction from one subset to the other. Moreover, by collapsing each subset into one vertex, we can form a meta-graph that corresponds to a directed cycle of three vertices. This small meta-graph describes the particular cluster-structure of this digraph.

\begin{figure}[!ht]
    \centering
    \includegraphics[width=0.4\textwidth]{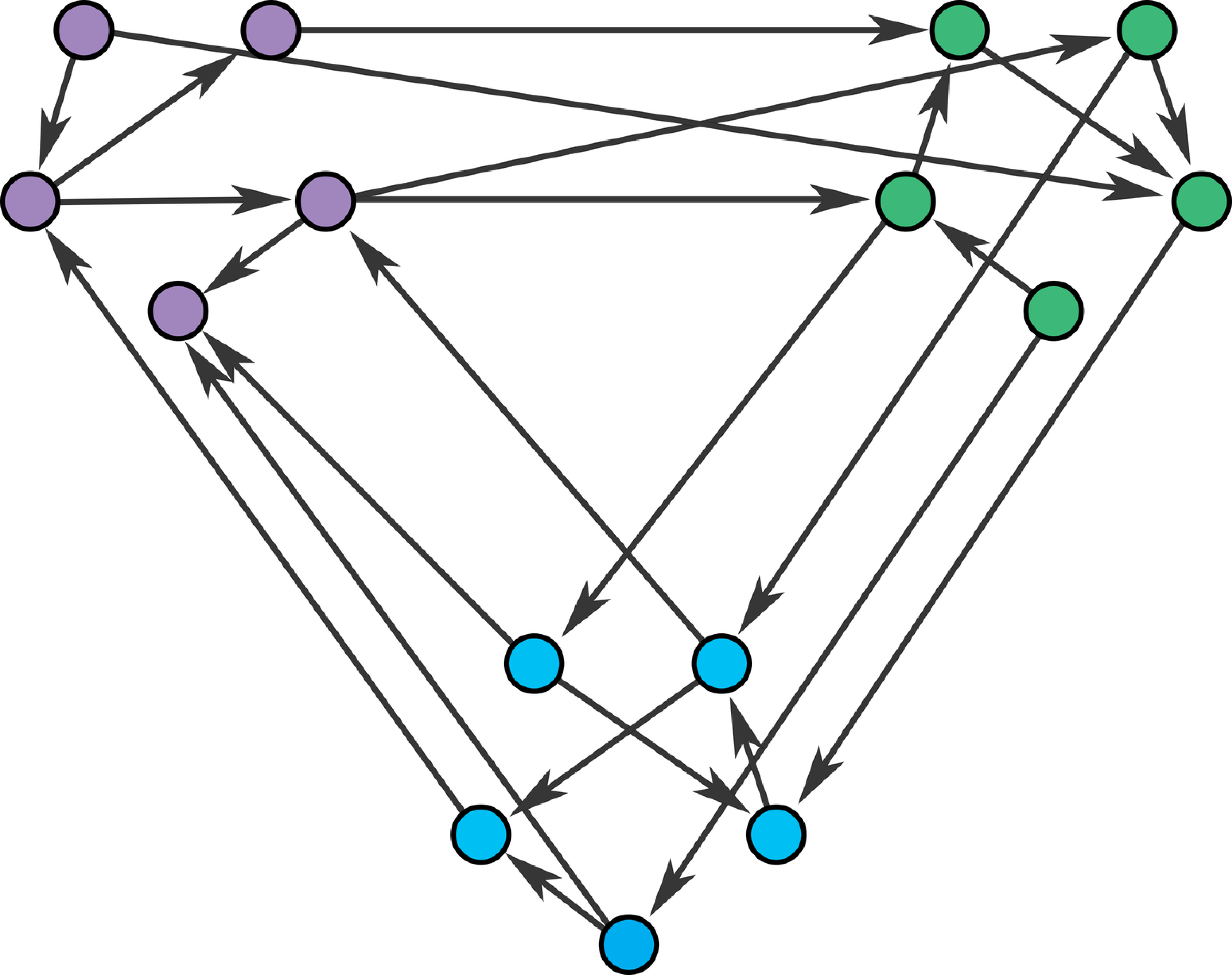}
    \caption{Example of a digraph with a cluster-structure that depends on the direction of the edges, rather than the density of the clusters.}
    \label{Fig:toyexample}
\end{figure}

To showcase the potential usefulness of this clustering task, consider a directed graph representing a food web, where vertices of the graph correspond to species of organisms present in a certain ecosystem, and there exists an edge from a vertex $u$ to a vertex $v$ if $u$ gets consumed by $v$. Graph-theoretic approaches have been applied to determine trophic levels \cite{mackay2020directed}, i.e., the position of the organisms in the food chain, essentially placing ecological communities on a linear ordering. In other words, these approaches compute a clustering with a meta-graph that corresponds to a directed path. Ecological communities, however, often have more complicated interactions that cannot be faithfully captured  by a linear ordering. Our clustering task aims to find the most relevant ecological communities together with their interactions, without imposing a path structure on the corresponding meta-graph. We will analyse a data set of this kind in Section~\ref{sec:real_exp}.

\subsection{Our approach}

We present a new algorithm for this clustering problem which is based on an iterative application of a particular spectral clustering procedure. Spectral clustering is a well known technique for clustering undirected graphs~\cite{ng2001spectral, ShiMalik}. It exploits the most significant eigenvectors of a matrix representation of a graph (such as the adjacency matrix or the Laplacian) to embed its vertices into a low dimensional Euclidean space. This embedding essentially transforms the original combinatorial problem into a geometric one; classical geometric clustering techniques such as $k$-means can then be applied on the embedded points to recover a partition of the vertex set into clusters. While spectral clustering algorithms are known to work well in the context of undirected graphs, their application to directed graphs is much less explored. One of the reasons for this scarcity of work is that traditional matrix representation of digraphs (such as the adjacency matrix) are not Hermitian and therefore do not admit a basis of orthonormal eigenvectors. Consequently, Cucuringu et al. \cite{cucuringu2020hermitian} have recently investigated a simple Hermitian representation of digraphs: given the traditional (asymmetric) adjacency matrix $W$ of a digraph, they consider the matrix $H = \mi W - \mi W^{\top}$, where $\mi$ is the imaginary unit. They show that spectral clustering with such a matrix representation produces clusters where the directions of the edges between clusters are characterised by strong imbalances. They provide some theoretical justification of the performance of their algorithm on the Directed Stochastic Block Model (DSBM), a random digraph model with a ground truth clustering. 

By investigating the quadratic form of the Hermitian matrix used by Cucuringu et al. (i.e., $x^* H x$), however, it becomes evident that their algorithm presents a major drawback: it favours cluster-structures whose meta-graph corresponds to a directed $4$-cycle. This limitation, which we explain in Section~\ref{sec:prelim}, became the motivation for  our work: can we find a matrix representation for digraphs that does not favour any particular fixed meta-graph, but rather, for each digraph, encodes the ``best'' meta-graph?

Given a clustering $\mathcal{S} = \{S_0,\dots,S_{k-1}\}$ for a digraph $G$, we construct a Hermitian matrix $H_{\mathcal{S}}$ whose quadratic form $x^* H_{\mathcal{S}} x$ penalise edges with a direction that does not follow the direction prescribed by the meta-graph associated with  $\mathcal{S}$. The entries of $H_{\mathcal{S}}$ are carefully defined using complex roots of unity in order to be able to penalise only the right edges. This representation crucially depends on both the graph $G$ and the clustering $\mathcal{S}$.  Thanks to this construction, we design an iterative algorithm that explores the space of possible meta-graphs by repeatedly applying spectral clustering to matrix representations corresponding to potentially different meta-graphs. The algorithm outputs the best clustering computed during these iterations according to some purposely designed measure.

We examine the performance of our algorithm on the DSBM and show it outperforms the algorithm by Cucuringu et al. for an extensive range of parameters. In particular, our algorithm seems especially more accurate when the meta-graph of the optimal clustering is sparse.
We also apply our algorithm to a diverse range of real-world data sets. For example, we analyse a data set related to the popular online card game Hearthstone and uncover interesting properties of its game design.

\subsection{Further related work}

There is a wealth of literature dedicated to graph partitioning and spectral clustering on undirected graphs~(see, e.g., \cite{LOT, ng2001spectral, PSZ, von2007tutorial}).
As observed by Malliaros and Vazirgiannis \cite{malliaros2013clustering}, when the graph is directed, the notion of a cluster often falls in one of two categories: \emph{density-based} or \emph{pattern-based}. Undirected clustering techniques typically follow the density-based notion, and many of these techniques can be extended to directed graphs \cite{chung05,gleich2006hierarchical, palla2007directed}.

We are instead interested in the pattern-based notion of clusters, where vertices are partitioned based on common patterns in the direction of the edges, rather than according to the edge density. 
Recently, several spectral algorithms for pattern-based clustering of digraphs have been proposed. These algorithms differ in the Hermitian representations they exploit. For example, 
Rohe et al. \cite{coclustering} consider symmetrisations $WW^\top$ and $W^\top W$ of the asymmetric adjacency matrix $W$. Satuluri and Parthasarathy \cite{bib-sym} consider instead the symmetrisation $WW^\top + W^\top W$.  The aforementioned Cucuringu et al. \cite{cucuringu2020hermitian} use the skew-symmetric Hermitian representation $\mi (W - W^\top)$ (and several normalisations of this matrix) and generally obtain better results than \cite{coclustering,bib-sym}. Laenen and Sun \cite{LaenenSun} use a variation of the representation by Cucuringu et al. where $\mi$ is replaced by the $k$-th root of unity $\omega_k$, with $k$ equal to the number of clusters. Applying techniques taken from \cite{PSZ, SZ19}, they obtain theoretical results for recovering clusters whose meta-graph is a directed $k$-cycle.

Finally, we mention that complex-valued Hermitian representations of digraphs have recently been applied to Graph Neural Networks as well~\cite{he2022msgnn,zhang2021magnet}.

\section{Notation and preliminaries}
\label{sec:prelim}
A (weighted) directed graph (or digraph) $G=(V,E,w)$ is a triple where $V$ is the set of vertices, $E \subseteq V \times V$ is the set of directed edges (which we denote with ordered pairs of vertices), and $w \colon V \times V \to \mathbb{R}_{\ge 0}$ is a nonnegative weight function such that $w(u,v)>0$ if and only if $(u,v) \in E$. An unweighted digraph $G=(V,E)$ is simply a digraph where the edges have unit weight. 
For simplicity, we will consider only simple graphs, which are graphs without cycles of length one or two.

Given two subsets of vertices $A,B \subset V$, not necessarily distinct, we denote with $E(A,B)$ the set of edges from $A$ to $B$, i.e., $E(A,B) = \{(u,v) \in E \colon u \in A, v \in B\}$.

We say $S \subseteq V$ is \emph{weakly connected} if, for any $x,y \in S$, there exists a sequence of vertices $x = u_0, u_1, \dots, u_l = y$, all  belonging to $S$, such that, for any $i=0,\dots,\ell-1$, $(u_i, u_{i+1}) \in E$ or $(u_{i+1}, u_i) \in E$. 

A \emph{clustering} of $G$ is a partition of the vertex set $V$ into $k \ge 2$ nonempty subsets. Given a clustering $\mathcal{S} = \{S_0,\dots,S_{k-1}\}$, we define the \emph{meta-graph} of $\mathcal{S}$ as the digraph $G^{\mathcal{S}} = (V^{\mathcal{S}},E^{\mathcal{S}})$ where $V^{\mathcal{S}} = \{0,\dots,k-1\}$ and $(i,j) \in E^{\mathcal{S}}$ if and only if most of the weight of edges in $G$ between $S_i$ and $S_j$ is oriented from $S_i$ to $S_j$; more precisely $(i,j) \in E^{\mathcal{S}}$ if and only if $\sum_{u \in S_i, v \in S_j} w(u,v) > \sum_{u \in S_j, v \in S_i} w(u,v)$.

The (asymmetric) adjacency matrix of $G$ is the matrix $W \in \mathbb{R}^{V \times V}$ such that $W_{u,v} = w(u,v)$ for any $u,v \in V$. We define the degree of a vertex $u \in V$ as the sum of the weights of the edges incident to $u$, i.e., $d(u) = \sum_{v \in V} (w(u,v) + w(v,u))$. The volume of a subset of vertices $S \subseteq V$ is the sum of the degrees of the vertices in $S$, i.e., $\vol{S} = \sum_{u \in S} d(u)$.  The degree matrix $D \in \mathbb{R}^{V \times V}$ is the diagonal matrix whose nonzero entries are equal to the degrees of the vertices, i.e., $D_{u,u} = d(u)$.

We call $M$ a Hermitian representation of $G=(V,E,w)$ if $M \in \mathbb{C}^{V \times V}$ such that $M = M^{*}$ and, for any $u,v \in V$, $M_{u,v} = (w(u,v) + w(v,u)) c_{u,v}$ where $c_{u,v} = \overline{c_{v,u}} \in \mathbb{C}$ and $|c_{u,v}|=1$. Given a Hermitian representation $M$ together with a degree matrix $D$, we define the Laplacian matrix as $L(M) = D-M$

The following lemma, which follows by simple computations, gives a useful formula for the quadratic form of the Laplacian $L(M)$. 

\begin{lemma}
\label{lem:quadform}
    Let $M$ be a Hermitian representation of $G=(V,E,w)$. Then, for any $x \in \mathbb{C}^V$, it holds that
    \[
        x^{*}L(M)x = \sum_{(u,v) \in E} w(u,v)\left|x_u - \frac{M_{u,v}}{w(u,v)} x_v\right|^2.
    \]
\end{lemma}

\begin{proof}
Recall that $|M_{u,v}|=w(u,v)$ if $(u,v) \in E$, $|M_{u,v}|=w(v,u)$ if $(v,u) \in E$, and $|M_{u,v}|=0$ otherwise. Moreover, $D_{u,u} = \sum_{v \in V} |M_{u,v}|$. Then, we have that
\begin{align*}
    x^* L(M) x
    &=
    \sum_{u \in V} D_{u,u} \overline{x_u} x_u - \sum_{u,v \in V}  M_{u,v} \overline{x_u} x_v
    \\ &=
    \sum_{u,v \in V} |M_{u,v}| \overline{x_u} x_u - \sum_{u,v \in V} M_{u,v} \overline{x_u} x_v
    \\ &=
    \sum_{(u,v) \in E} |M_{u,v}| \overline{x_u} x_u + \sum_{(v,u) \in E} |M_{u,v}| \overline{x_u} x_u
    - \sum_{(u,v) \in E} M_{u,v} \overline{x_u} x_v - \sum_{(v,u) \in E} M_{u,v} \overline{x_u} x_v
    \\ &=
    \sum_{(u,v) \in E} |M_{u,v}| \overline{x_u} x_u + \sum_{(u,v) \in E} |M_{v,u}| \overline{x_v} x_v 
    - \sum_{(u,v) \in E} M_{u,v} \overline{x_u} x_v - \sum_{(u,v) \in E} M_{v,u} \overline{x_v} x_u
    \\ &=
    \sum_{(u,v) \in E} w(u,v) \left( \overline{x_u} x_u + \overline{x_v} x_v 
    - \frac{M_{u,v}}{w(u,v)} \overline{x_u} x_v - \frac{\overline{M_{u,v}}}{w(u,v)} \overline{x_v} x_u \right)
    \\ &=
    \sum_{(u,v) \in E} w(u,v) \left|x_u - \frac{M_{u,v}}{w(u,v)}x_v \right|^2.
\end{align*}
\end{proof}

This simple lemma is crucial to understand the algorithm by Cucuringu et al. \cite{cucuringu2020hermitian}. In fact, their algorithm  uses the eigenvectors corresponding to the smallest $k$ eigenvalues of $L(M) = D-M$ where $M=\mi W -\mi W^{\top}$. By the Courant-Fischer theorem \cite[Theorem 4.2.11]{HJ} and Lemma~\ref{lem:quadform}, the eigenvector corresponding to the smallest eigenvalue of $L(M)$ can be seen as the solution to the following minimisation problem:
\begin{equation}
\label{eq:quadi}
    \min_{x \in \mathbb{C}^V \setminus \{0\}} \sum_{(u,v)\in E} \frac{w(u,v)\left|x_u - \mi x_v\right|^2}{x^{*}x}.
\end{equation}
Notice this quadratic form is  equal to zero if and only if $x_u = \mi x_v$ for all $(u,v) \in E$. 
For that to be the case, there must be a $c \in \mathbb{C}$ such that, for all $u \in V$ belonging to the same weakly connected component, $c \cdot x_u \in \{1, \mi, -1, -\mi\}$. After this rotation and scaling, the eigenvector $x$ must map the vertices to powers of the fourth root of unity and edges must be oriented from vertices corresponding to a power of $\mi$ to the next.
Therefore, $L(M)$ has a zero eigenvalue if and only if we can embed $G$ into a directed $4$-cycle respecting all the edge directions. This tells us that by solving the minimisation problem (\ref{eq:quadi}), we are trying to obtain a clustering whose meta-graph is cyclic. This does not imply a clustering obtained by using the minimiser of (\ref{eq:quadi}) will necessarily have a cyclic meta-graph, but rather that the bottom eigenvectors of $L(M)$ implicitly favour such a cyclic structure.

\section{Our algorithm}
In this section we present our algorithm for clustering directed graphs. The starting point for designing our algorithm is to construct a matrix representation $M^{\mathcal{S}}$ of a digraph $G$ that, given a clustering $\mathcal{S}$, penalises edges whose direction does not follow the one prescribed by the meta-graph of $\mathcal{S}$. Each iteration will try to refine the clustering and its corresponding meta-graph, such that the matrix $M^{\mathcal{S}}$ is dependent on the clustering from the previous step. 

To construct $M^{\mathcal{S}}$, let us assign to each cluster a power of the $k$-th root of unity, e.g., the cluster $S_i$ is assigned to $\omega_k^i = \exp\left(\frac{2i \pi \mi}{k}\right)$. We can represent this assignment as a vector $x \in \mathbb{C}^V$ such that $x_u = \omega_k^j$ whenever $u \in S_j$ for any $j \in \{0,\dots,k-1\}$.  Let $u \in S_i$ and $v \in S_j$ and suppose $(u,v) \in E$. Consider the term in the quadratic form $x^{*}L(M^{\mathcal{S}})x$ associated with $(u,v)$, i.e., $w(u,v) \cdot \left|x_u - \frac{M^{\mathcal{S}}_{u,v}}{w(u,v)} x_v\right|^2$. We would like this term to be equal to zero if there exists an edge oriented from $i$ to $j$ in the meta-graph $G^{\mathcal{S}}$, otherwise we would like this term to be equal to $w(u,v)$. Indeed, by the definition of $x$, it holds that
\begin{align}
\label{eq:pen=1}
\left|x_u - \frac{M^{\mathcal{S}}_{u,v}}{w(u,v)} x_v\right| = 1 &\iff M^{\mathcal{S}}_{u,v} = w(u,v) \cdot \omega_k^{i-j} \cdot e^{\pm \pi \mi / 3}, \\
\label{eq:pen=0}
\left|x_u - \frac{M^{\mathcal{S}}_{u,v}}{w(u,v)} x_v\right| = 0 &\iff M^{\mathcal{S}}_{u,v} = w(u,v) \cdot \omega_k^{i-j}.
\end{align}

%
%
To understand why \eqref{eq:pen=1} holds, we first prove the forward direction.
Let $\frac{M^{\mathcal{S}}_{u,v}}{w(u,v)} = e^{c \mi}$, where $u \in S_i$, $v \in S_j$ and $c \in \mathbb{R}$. Then
\begin{align*}
    \left|x_u - \frac{M^{\mathcal{S}}_{u,v}}{w(u,v)} x_v\right| = 1 
    \iff 
    \left|\omega_k^i - e^{c \mi} \omega_k^j \right| &= 1
    \iff 
    \left| 1 - e^{\left(c + \frac{2\pi(j-i)}{k} \right) \mi} \right| = 1.
\end{align*}
Let $z \in \mathbb{R}$. Using Euler's formula $e^{z \mi} = \cos(z) + \mi \sin(z)$, we have that
\[
\left| 1 - e^{z \mi} \right| = 1 \iff \cos(z) = \frac{1}{2} \iff z = \pm \frac{\pi}{3} + 2n\pi,
\]
for any $n \in \mathbb{Z}$. 
Therefore,
\[
\frac{2\pi(j-i)}{k}+c = \pm \frac{\pi}{3} + 2n\pi
\]
and we can define $M^{\mathcal{S}}_{u,v}$ as
\[
M^{\mathcal{S}}_{u,v} = w(u,v) \cdot \omega_k^{i-j} \cdot e^{\pm \pi \mi / 3},
\]
where the plus or minus sign in the exponent can be chosen arbitrarily.

To prove the reverse direction, we have
\[
M^{\mathcal{S}}_{u,v} = w(u,v) \cdot \omega_k^{i-j} \cdot e^{\pm \pi \mi / 3}
 \iff
\frac{M^{\mathcal{S}}_{u,v}}{w(u,v)} \omega_k^j =  \omega_k^i \cdot e^{\pm \pi \mi / 3}.
\]
Hence
\[
\left| \omega_k^i - \frac{M^{\mathcal{S}}_{u,v}}{w(u,v)} \omega_k^j \right| =  \left| \omega_k^i \right| \left|1 -  e^{\pm \pi \mi / 3} \right| = 1.
\]

The proof of \eqref{eq:pen=0} follows using similar arguments.


With these considerations in mind, we define the Hermitian matrix $M^{\mathcal{S}}$ as follows. Let $u \in S_i$ and $v \in S_j$. Then

\begin{equation}
\label{eq:Mdef}
    M_{u,v}^{\mathcal{S}} =
    \begin{cases}
    \begin{aligned}
    & w(u,v) \cdot \omega_k^{i-j} && \text{if } (u,v) \in E \text{ and } (i,j) \in E^{\mathcal{S}} \\
    & w(v,u) \cdot \omega_k^{j-i} && \text{if } (v,u) \in E \text{ and } (i,j) \in E^{\mathcal{S}} \\
    & w(u,v) \cdot \omega_k^{i-j} \mathrm{e}^{\pi \mi / 3} && \text{if } (u,v) \in E, (i,j) \not \in E^{\mathcal{S}} \text{ and } i \neq j \\
    & w(v,u) \cdot \omega_k^{j-i} \mathrm{e}^{-\pi \mi / 3}&& \text{if } (v,u) \in E,  (i,j) \not \in E^{\mathcal{S}} \text{ and } i \neq j \\
    &w(u,v) && \text{if } (u,v) \in E \text{ and } i = j \\
    &w(v,u) && \text{if } (v,u) \in E \text{ and } i = j \\
    &0 && \text{if } (u,v) \not\in E \text{ and } (v,u) \not\in E.
    \end{aligned}
    \end{cases}
\end{equation}

We remark that \eqref{eq:pen=1} allows for two possible alternative choices for the third and fourth line in \eqref{eq:Mdef}. We have arbitrarily chosen one of them.
Also notice that the penultimate and antepenultimate conditions cover the situation in which $i=j$, meaning we do not assign penalties on edges inside clusters. This can be changed by instead using the matrix
\begin{equation}
    \label{eq:Mdef_pen}
    {M_P^{\mathcal{S}}}_{u,v} =
    \begin{cases}
    \begin{aligned}
    & w(u,v) \cdot \omega_k^{i-j} && \text{if } (u,v) \in E \text{ and } (i,j) \in E^{\mathcal{S}} \\
    & w(v,u) \cdot \omega_k^{j-i} && \text{if } (v,u) \in E \text{ and } (i,j) \in E^{\mathcal{S}} \\
    & w(u,v) \cdot \omega_k^{i-j} \mathrm{e}^{\pi \mi / 3} && \text{if } (u,v) \in E, (i,j) \not \in E^{\mathcal{S}} \text{ and } i \neq j \\
    & w(v,u) \cdot \omega_k^{j-i} \mathrm{e}^{-\pi \mi / 3}&& \text{if } (v,u) \in E,  (i,j) \not \in E^{\mathcal{S}} \text{ and } i \neq j \\
    & w(u,v) \cdot e^{\frac{\mi \pi}{3}} && \text{if } (u,v) \in E \text{ and } i = j \\
    & w(u,v) \cdot e^{\frac{-\mi \pi}{3}} && \text{if } (v,u) \in E \text{ and } i = j \\
    &0 && \text{if } (u,v) \not\in E \text{ and } (v,u) \not\in E.
    \end{aligned}
    \end{cases}.
\end{equation}
In this case, each edge inside a cluster will be penalised, which can be beneficial in certain applications and we explore this variant in Section~\ref{sec:real_exp}. Unless otherwise stated, we will use the matrix representation defined in \eqref{eq:Mdef}. We remark that also the algorithm by Cucuringu et al. implicitly penalises edges inside clusters. This is because for an edge $(u,v) \in E$, where $u$ and $v$ are in the same cluster, we expect $x_u = x_v$, leading to a non-zero penalty of $|x_u - \mi x_v| = \sqrt{2}$.

Clearly, $M^\mathcal{S}$ is Hermitian by construction. Furthermore, the Laplacian $L(M^{\mathcal{S}})$, which is positive semidefinite, has a zero eigenvalue if and only if $\mathcal{S}$ is a \emph{perfect} clustering, in the sense that the directions of all edges between clusters respect the orientations given by the meta-graph $G^{\mathcal{S}}$. More precisely, we say $\mathcal{S}$ is perfect if, 
for any $u \in S_i, v \in S_j$ with $i \ne j$, it holds that $(u,v) \in E \implies (i,j) \in E^{\mathcal{S}}$.  
The next lemma formalises the fact above. 

\begin{lemma}
    \label{lem:zeroeig}
    Let $\mathcal{S}=\{S_0,\dots,S_{k-1}\}$ be a clustering and $M^{\mathcal{S}}$ defined as in (\ref{eq:Mdef}).  Then, if $\mathcal{S}$ is a perfect clustering, $L(M^{\mathcal{S}})$ has a zero eigenvalue. Furthermore, if both each cluster and the meta-graph of $\calS$ are weakly connected, then $L(M^{\mathcal{S}})$ having a zero eigenvalue implies $\calS$ is perfect.
\end{lemma}

\begin{proof}
We 
first show that, if $\calS$ is perfect, then $L(M^\calS)$ has a zero eigenvalue. We prove this by constructing $x \in \mathbb{C}^V \setminus \{0\}$ such that $x^{*} L(M^{\mathcal{S}}) x =0$. 

Let $E_1=\{(u,v)\in E: u \in S_i, v \in S_j, (i,j)\in E_\mathcal{S}\}$ and $E_2=\{(u,v)\in E: u \in S_i, v \in S_j, i=j \}$. Since we assume  $\calS$ is perfect, then we have $E=E_1 \cup E_2$.

Take $x$ such that $x_u = \omega_k^i$ for all $u \in S_i$, for $i = 1,\dots,k$. Then, using Lemma~\ref{lem:quadform}, we have that
\begin{align*}
    x^{*} L(M^\mathcal{S}) x
    &=
    \sum_{(u,v) \in E} w(u,v)\left|x_u - \frac{M^{\calS}_{u,v}}{w(u,v)} x_v\right|^2
    \\
    &=
    \sum_{(u,v) \in E_1} w(u,v)\left|x_u - \frac{M^{\calS}_{u,v}}{w(u,v)} x_v\right|^2
    + \sum_{(u,v) \in E_2} w(u,v)\left|x_u - \frac{M^{\calS}_{u,v}}{w(u,v)} x_v\right|^2
    \\
    &=
    \sum_{(u,v) \in E_1} w(u,v)\left| \omega_k^i - \omega_k^{i-j}  \omega_k^j\right|^2 
    + \sum_{(u,v) \in E_2} w(u,v)\left| \omega_k^i -  \omega_k^i\right|^2
    \\
    &=
    0.
\end{align*}
This proves the first statement.

Next, we prove that, under weak connectivity of the clusters and the meta-graph, $L(M^\calS)$ having a zero eigenvalue implies $\calS$ is perfect.

Suppose by contradiction that $\mathcal{S}$ is not perfect but $L(M^\calS)$ has a zero eigenvalue. Then, by Lemma~\ref{lem:quadform}, there exists $ x \in \mathbb{C}^V \setminus \{0\}$ such that, for all $(u,v)\in E$, it holds that
$\left|x_u - \frac{M^{\calS}_{u,v}}{w(u,v)} x_v\right| =0$. Moreover, since $\mathcal{S}$ is not perfect, 
two situations can occur: either (1)  there exists $(u,v),(v',u') \in E$ such that $u,u'\in S_i$, $v,v'\in S_j$, and $(i,j) \in E^\calS$; or (2) there exists $(u,v),(v',u') \in E$ such that $u,u'\in S_i$, $v,v'\in S_j$, and $(i,j),(j,i) \not\in E^S$, with $i \ne j$ (we remark this can only happen when exactly half of the edge weight between $S_i$ and $S_j$ is oriented from the former to the latter).

Assume (1) holds. Then, by the weak connectivity of the clusters, there must be a sequence of vertices $u = u_0, u_1, \dots, u_\ell = u'$ such that $u_r \in S_i$ and either $(u_r,u_{r+1}) \in E$ or  $(u_{r+1},u_r) \in E$ for any $r=0,\dots,\ell-1$. By the definition of $M^{\calS}$, we have that $\left|x_{u_r} -  x_{u_{r+1}}\right| =0$ for any $r=0,\dots,\ell-1$, which implies $x_u = x_{u'}$. Similarly, $x_v = x_{v'}$. However, since $(i,j) \in E^\calS$, we have that 
\begin{equation*}
    \left|x_{u} - \frac{M^{\calS}_{u,v}}{w(u,v)} x_{v}\right|
    =
    \left|x_{u} - \omega_k^{i-j}  x_{v}\right| = 0,
\end{equation*}
which implies $x_u = \omega_k^{i-j}  x_{v}$. Similarly, it must hold that
\begin{equation*}
    \left|x_{v'} - \frac{M^{\calS}_{v',u'}}{w(v',u')} x_{u'}\right|
    =
    \left|x_{v'} - \omega_k^{j-i} \mathrm{e}^{\pi \mi / 3} x_{u'}\right| = 0,
\end{equation*}
and therefore $x_{v'} = \omega_k^{j-i} \mathrm{e}^{\pi \mi / 3} x_{u'}$.
But since $x_u = x_{u'}$ and $x_v=x_{v'}$, we reach a contradiction.


We now assume (2) holds and use similar arguments as above. The key difference in this case is that as $(i,j),(j,i) \notin E^\mathcal{S}$, then both $(u,v) \in E$ and $(v',u') \in E$ incur a penalty. Under the weak connectivity of the clusters, we have $x_w = x_{w'}$ for all $w,w'$ in the same cluster. Furthermore, if $(a,b) \in E^\mathcal{S}$, then $x_w = \omega_k^{a-b}x_{w'}$ for any $w \in S_{a}$, $w' \in S_{b}$. Under the weak connectivity of the meta-graph, this can be iteratively applied to an underlying undirected path between $u$ and $v$ that contains only edges which do not incur a penalty (thus not including any edges directly between $S_i$ and $S_j$) to give $x_u = \omega_k^{i-j} x_v$. However, since $(i,j) \not\in E^S$ and $i \neq j$, then $M^{\calS}_{u,v} = w(u,v) \omega_k^{i-j}\mathrm{e}^{\pi \mi / 3}$, hence it must also hold that
\begin{equation*}
    \left|x_{u} - \frac{M^{\calS}_{u,v}}{w(u,v)} x_{v}\right|
    =
    \left|x_{u} - \omega_k^{i-j} \mathrm{e}^{\pi \mi / 3} x_{v}\right| = 0.
\end{equation*}
This also leads to a contradiction.
\end{proof}

The weak connectivity conditions in the lemma are there essentially to prevent the situation where disconnected components in either the clusters or the meta-graph allows a vector $x$ to be zero on a portion of the graph without incurring any penalty. 

An important feature of the matrix construction of (\ref{eq:Mdef}) is that we are able to encode \emph{any} possible meta-graph, not just cyclic  meta-graphs as in \cite{cucuringu2020hermitian,LaenenSun}. The problem with this construction is that it requires to know in advance a good clustering of the graph, which is exactly what we would like to compute. Notice, however, that if a clustering $\mathcal{S}$ on a given iteration is poor, in the sense that many edges do not respect the meta-graph $G^\mathcal{S}$ constructed using $\mathcal{S}$, then a vector $x$ encoding the clustering will result in a large quadratic form $x^{*}L(M^\mathcal{S})x$. Therefore, it is very likely there will be a different vector $y$ with smaller quadratic form $y^{*}L(M^\mathcal{S})y$. The hope is that we can use such a vector to obtain a new clustering where a larger proportion of edges respect the corresponding meta-graph. 

These considerations suggest an iterative algorithm for digraph clustering. Start with an arbitrary clustering $\mathcal{S}_0$. Construct the corresponding matrix representation $M^{\mathcal{S}_0}$. Use its eigenvectors to compute a new clustering $\calS_1$, and repeat for a number $T$ of iterations. At the end, we will have computed $T$ candidate clusterings: we can simply output the best one according to some measure $\delta$ that we will define later. The formal description of our algorithm is presented in Algorithm~\ref{alg:ours}.

\begin{algorithm}[!ht]
\caption{Iterative Spectral Clustering for Digraphs}
\label{alg:ours}

\begin{flushleft}
\textbf{Input:}
digraph $G=(V,E,w)$, number of clusters $k \in \mathbb{N}$, number of iterations $T \in \mathbb{N}$, initial clustering $\calS_0 = \{S_0, \dots, S_{k-1}\}$.

\textbf{Output:} 
clustering $\mathcal{S_*}$.

\textbf{Algorithm:}
\begin{algorithmic}[1]

\For{$t=1,\dots,T$}

\State Compute matrix representation $M^{\mathcal{S}_{t-1}}$.

\State Compute the normalised matrix representation $\sym{M}^{\mathcal{S}_{t-1}} = D^{-1/2} M^{\mathcal{S}_{t-1}} D^{-1/2}$.

\State Compute the orthonormal eigenvectors $f_{1}, \dots, f_{k}$ corresponding to the $k$ eigenvalues with largest absolute value of $\sym{M}^{\mathcal{S}_{t-1}}$.

\State Construct $F \in \mathbb{R}^{|V| \times k}$ whose columns are the eigenvectors $f_{1}, \dots, f_{k}$.

\State Compute $H=D^{-1/2}F$.

\State Use $k$-means clustering on the rows of $H$ to obtain a new clustering partition $\mathcal{S}_{t}$.

\EndFor

\State Return $\calS_{t^*}$ where $t^*= \argmin{0 \leq t \leq T} \; \delta(\mathcal{S}_t)$.

\end{algorithmic}
\end{flushleft}
\end{algorithm}

A few observations regarding Algorithm~\ref{alg:ours} are in order. Firstly, notice that instead of working with the matrix representation $M^{\mathcal{S}_{t-1}}$, we use $\sym{M}^{\mathcal{S}_{t-1}} = D^{-1/2} M^{\mathcal{S}_{t-1}} D^{-1/2}$. This normalisation is standard practice in spectral clustering algorithms \cite{von2007tutorial} and is needed to prevent vertices with very large degree dominating the top eigenvectors of the matrix. The normalisation $H = D^{-1/2}F$ from line $5$ in Algorithm~\ref{alg:ours} is also standard practice in spectral clustering algorithms and required because $F$ tends to depend on the square root of the degree of each vertex. Secondly, instead of using just the top eigenvector of $\sym{M}^{\mathcal{S}_{t-1}}$, we use the $k$ eigenvectors corresponding to the top $k$ eigenvalues in absolute value. We do this because, in practice, it has resulted in better experimental results. Finally, the initial clustering $\calS_0$ can be chosen arbitrarily; in our experiments we simply generate $\calS_0$ at random.

To finish formalising Algorithm~\ref{alg:ours}, we need to define a measure $\delta$ that quantifies how good a clustering is. We would like a measure that penalises a clustering whenever many edges in the graph are not oriented according to the meta-graph corresponding to the clustering partition. Moreover, we would like the amount of penalty assigned to each edge to depend on the volume of the cluster(s) incident to that edge: edges incident to smaller clusters should be penalised more. With these considerations in mind, we define the measure $\delta$ as follows.
\begin{definition}
    Let $\mathcal{S} = \{S_0,\dots,S_{k-1}\}$ be a clustering of $G=(V,E,w)$ with meta-graph $G^\calS = (V^\calS, E^\mathcal{S})$. Then, the clustering value of $\mathcal{S}$ is
    \[
        \delta(\mathcal{S}) = \sum_{(i,j) \in E^\mathcal{S}} \sum_{(u,v) \in E(S_j, S_i)} \frac{w(u,v)}{\min\{ \vol{S_i},\vol{S_j}\}}.
    \]
\end{definition}

This cost function is similar in spirit to the measure of \emph{Flow Imbalance} defined by Cucuringu et al. \cite{cucuringu2020hermitian} and can be considered a generalisation of the \emph{Flow Ratio} of Laenen and Sun \cite{LaenenSun} from cyclic cluster-structures to arbitrary meta-graphs. The scaling by minimum volume relates to the concept of conductance of a cut within a network~\cite{LOT}, which favours partitions of clusters of similar size.

Notice that $\delta(\mathcal{S})$ is between zero and $k$. In particular, it is equal to zero if and only if $\calS$ is a perfect clustering. Moreover, $\delta(\calS)$ is related to the quadratic form of $L(M^\mathcal{S})$. When $S_0,\dots,S_{k-1}$ all have the same volume, the smallest eigenvalue of $L(M^\mathcal{S})$ is, in a way, a convex relaxation for $\delta(\mathcal{S})$. Indeed, let $x \in \mathbb{C}^V$ such that $x_u = \omega_k^i$ if $u \in S_i$. Then, $x^* L(M^\mathcal{S}) x = \vol{S_0} \cdot \delta(\mathcal{S})$. 

One potential weakness of $\delta$ is that it does not penalise edges inside the clusters. This might bias towards clusters with a high inner-density and a small boundary. While we have not found evidence of this behaviour for our algorithm, it is still worth defining a modified clustering value that penalises edges inside clusters. Therefore, we define an alternative clustering value as follows
\[
        \delta_P(\mathcal{S}) = \sum_{(i,j) \not\in E^\mathcal{S}} \sum_{(u,v) \in E(S_i, S_j)} \frac{w(u,v)}{\min\{ \vol{S_i},\vol{S_j}\}},
\]
where the first sum includes also all $i=j$. 
We generally use this alternative clustering value in conjunction with the matrix representation ${M_P^{\mathcal{S}}}_{u,v}$ defined in \eqref{eq:Mdef_pen} that penalises edges inside clusters. 

Finally, in both variants, by scaling the contribution of each edge with respect to the minimum volume between the cluster(s) containing the endpoints of the edge, a greater penalty is applied to edges incident to small clusters. This is to penalise a clustering based on the proportion of edges between clusters that do not conform to the direction described by the meta-graph. If a clustering has very small clusters, however, penalising by proportion of edges implies most of the penalty will come from the smaller clusters, which might result in finding large clusters that are not particularly relevant. We therefore emphasise that the definition of clustering value can be tweaked based on the application to recover clusters that better align with the aim of the application.

\section{Experimental results}
In this section we present experimental results for our algorithm both on synthetic and real-world data sets. We will compare our results with the algorithm by Cucuringu et al.~\cite{cucuringu2020hermitian} (more precisely their random-walk normalisation), which we call the CLSZ algorithm for brevity. We call CLSZ the \textit{state-of-the-art}, as it has the most similar objective to our own while yielding an improvement over numerous alternative digraph clustering algorithms. In our experiments, we choose a number of iterations $T$ such that results are easy to visualise, while further iterations do not yield a significant improvement.

\subsection{The Directed Stochastic Block Model}
\label{sec:dsbm}
The Directed Stochastic Block Model (DSBM) is a random model of digraphs with a ground truth clustering. It is a natural directed generalisation of the traditional Stochastic Block Model. We now explain how digraphs from the DSBM are sampled.

We assume the ground truth clustering contains $k$ clusters $\mathcal{T} = \{S_0,\dots, S_{k-1}\}$ each of size $n$, for a total of $N=n \cdot k$ vertices in the digraph. We sample a digraph $G=(V,E)$ in the following way. First, we generate a meta-graph $G^\mathcal{T} = (\{0,\dots,k-1\}, E^\mathcal{T})$ associated with this ground-truth clustering. We do this by sampling, for any $i,j \in \{0,\dots,k-1\}$, the undirected edge $\{i,j\}$ with probability $\gamma \in [0,1]$. Then, for any sampled undirected edge $\{i,j\}$, we add with probability $1/2$ the directed edge $(i,j)$ to the meta-graph $G^\mathcal{T}$, otherwise we add $(j,i)$. The parameter $\gamma$ can be used to control the \emph{sparsity} of our meta-graph $G^\mathcal{T}$. We now sample edges in the graph $G$ according to the sampled meta-graph $G^\mathcal{T}$. First we start sampling edges inside clusters: for any $i \in \{0,\dots, k-1\}$ and any $u,v \in S_i$, with probability $p \in [0,1]$, we add an edge between $u$ and $v$, with orientation chosen uniformly at random. Then, we sample edges between clusters: for any $(i,j) \in E^\mathcal{T}$, and any $u \in S_i, v \in S_j$, we add an edge between the two with probability $p$; the edge will be oriented from $u$ to $v$ with probability $\eta \in [1/2,1]$, from $v$ to $u$ with probability $1-\eta$. Essentially $\gamma$ represents the sparsity of the meta-graph $G^\mathcal{T}$, $p$ the sparsity of the digraph $G$, and $\eta$ how dominant are the edge directions between clusters (alternatively, $1-\eta$ can be thought of as the noise level).

Variants of this model can be created, for example, by allowing edges between clusters that are not connected in the meta-graph, or by setting different inter- and intra-cluster densities. We decided, however, to define a model with fewer parameters that would still capture the varying complexity of our clustering task. We remark a very similar model has been previously used by \cite{coclustering, cucuringu2020hermitian}.


To highlight the performance of our algorithm on the DSBM, we first look at a graph sampled from the DSBM with parameters $n = 100$, $k = 5$, $\gamma = 0.4$, $p = 0.5$ and $\eta = 0.6$. Since digraphs sampled from the DSBM possess a ground-truth clustering, we can measure the quality of a clustering using the misclassification error, i.e., the proportion of vertices classified incorrectly. This is done by applying the Hungarian algorithm to find the permutation of the clusters that minimises the symmetric difference between our clustering and the ground truth.

\begin{figure}[!ht]
    \centering
    \includegraphics[width=0.35\textwidth]{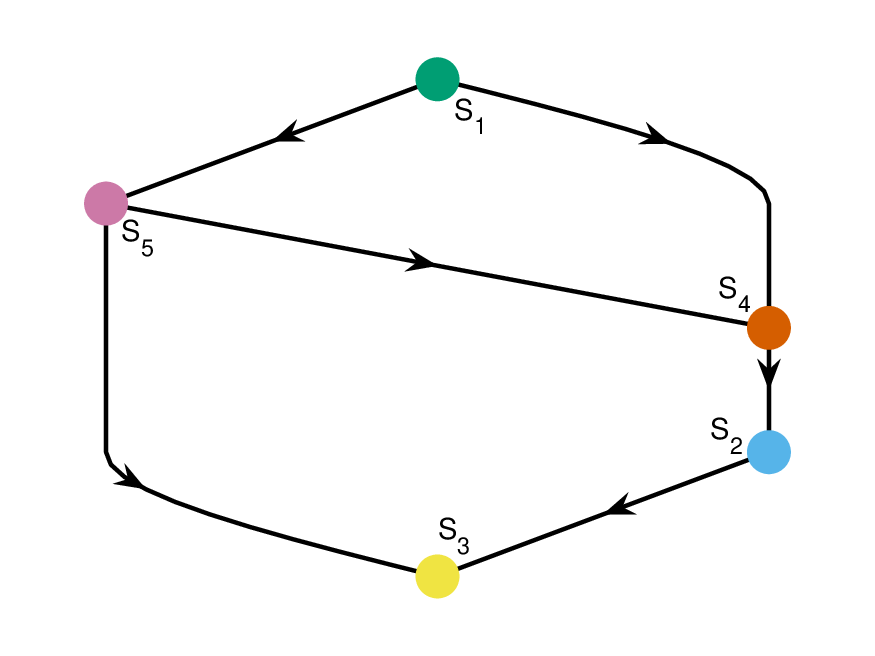}
    \includegraphics[width=0.35\textwidth]{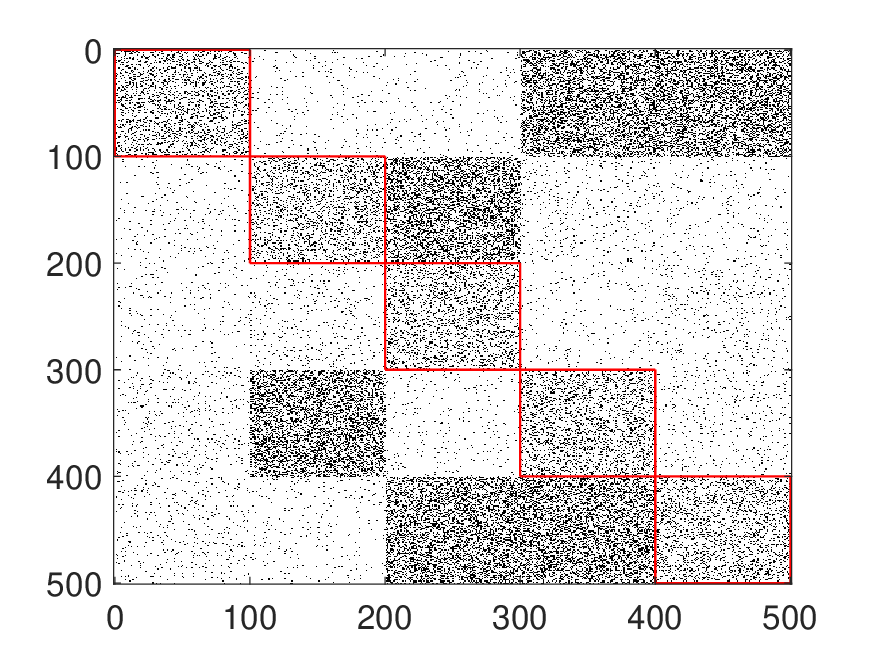}
    \caption{DSBM with parameters $n = 100$, $k = 5$, $\gamma = 0.4$, $p = 0.5$ and $\eta = 0.6$. Meta-graph used to sample the graph (left). Heat map of the adjacency matrix, reordered with respect to the recovered clustering shown in red (right).}
    \label{Fig:DSBM_traj}
\end{figure}

\begin{figure}[!ht]
    \centering
    \includegraphics[width=0.7\textwidth]{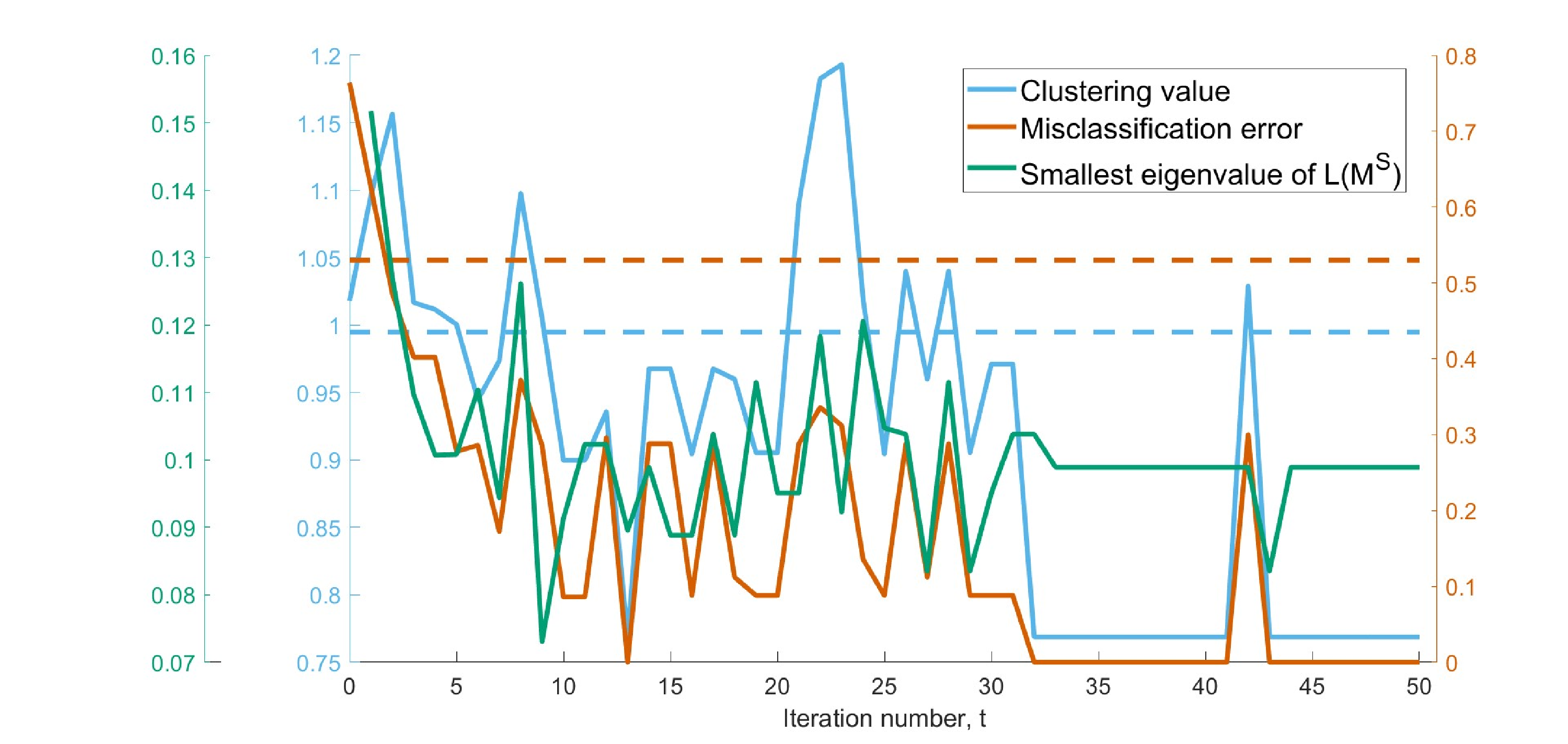}
    \caption{DSBM with parameters $n = 100$, $k = 5$, $\gamma = 0.4$, $p = 0.5$ and $\eta = 0.6$. Comparison of the performance of our algorithm (solid lines) vs CLSZ (dashed lines). 
    \newline
    Our minimum clustering value $\delta$: 0.769. CLSZ  clustering value $\delta$: 0.995.
    \newline
    Our minimum misclassification error: 0.000. CLSZ misclassification error: 0.530.}
    \label{fig:DSBM_lambda}
\end{figure}

Fig.~\ref{Fig:DSBM_traj} (left) shows the meta-graph used to generate the digraph. 
Fig.~\ref{Fig:DSBM_traj} (right) shows the 
heat map of the adjacency matrix of the digraph $G$ generated from the DSBM.
In Fig.~\ref{fig:DSBM_lambda} we show the behaviour of our algorithm on the same digraph, together with a comparison with the CLSZ algorithm. First, we highlight that the misclassification error (solid red), i.e., the proportion of vertices misclassified, generally decreases with time and indeed we are able to recover a clustering of zero error. On the other hand, the CLSZ algorithm performs almost as bad as the random clustering our algorithm started with, both from the point of view of the misclassification error (dashed red) and the clustering value $\delta$ (dashed blue). We also emphasise how the clustering value $\delta(\calS_t)$ tracks fairly well the misclassification error of $\calS_t$. This is true also for the smallest eigenvalue of $L(M^{\calS_t})$, even though not as well. We remark that the clustering value $\delta(\calS_t)$, however, does not reach zero even when a clustering with zero misclassification error has been found. This is due to the fact that, as is shown in  
Fig.~\ref{Fig:DSBM_traj} (right), the digraph is characterised by a nontrivial amount of noise. Finally, we also observe that unfortunately none of these three measures are monotonically decreasing with $t$, raising the more complex question of which iteration the algorithm should terminate at. We believe this behaviour is due to the fact that not always the top $k$ eigenvectors of $\sym{M}^{\calS_t}$ are all significant, but rather some of them might be mainly noise. Injecting this noise, however, prevents the algorithm to get stuck in a local minima.



\begin{figure}[!ht]
    \centering
    \includegraphics[width=0.4\textwidth]{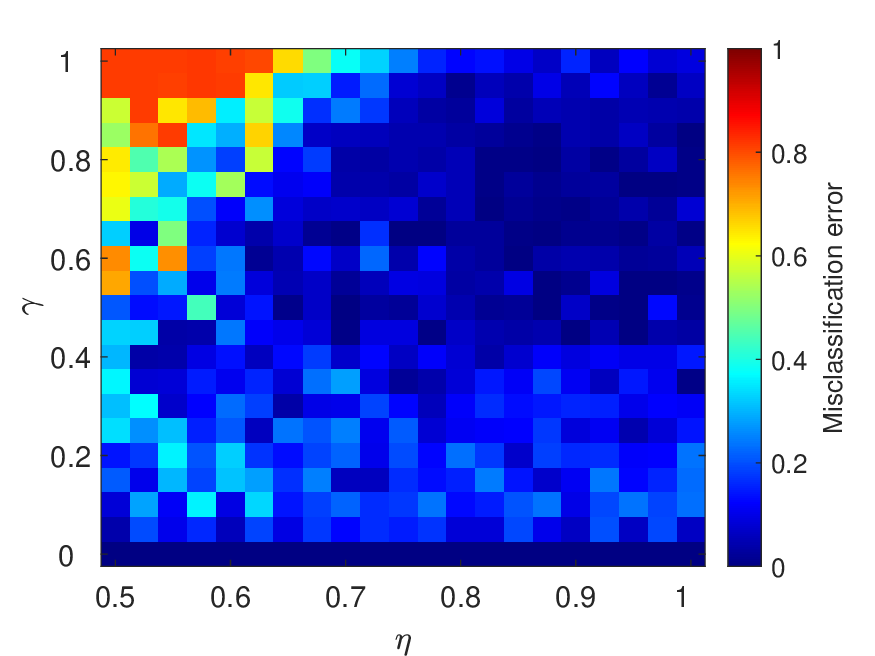}
    \label{Fig:DSBM_grid_luca}
    \includegraphics[width=0.4\textwidth]{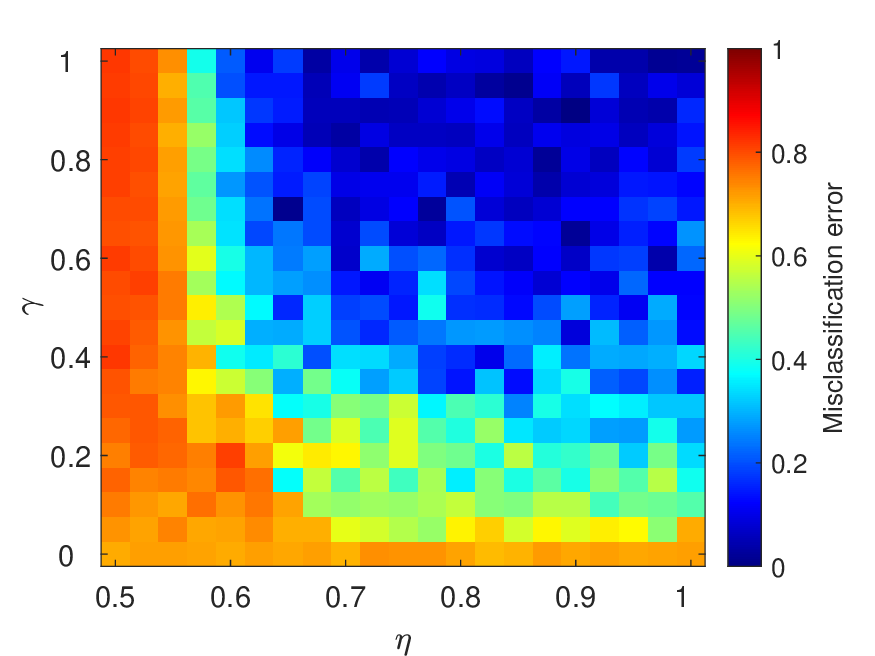}
    \caption{A comparison between our algorithm (left), with $T=50$, and CLSZ (right) on DSBM with fixed parameters $n=100$, $k=7$, $p=0.8$ and varying  $\eta$ and $\gamma$. Colour represents the average error over $10$ independent simulations. }
    \label{Fig:DSBM_grid_our}
\end{figure}

Fig.~\ref{Fig:DSBM_grid_our} presents a more extensive comparison on the performances of our algorithm and CLSZ. In particular, we compare their performances for different values of $\gamma$, the sparsity of the meta-graph, and $\eta$, the noise level in the direction of the edges. It is clear our algorithm outperforms CLSZ in the majority of the parameters' range, especially when the meta-graph is sparse (small values of $\gamma$). In particular, notice that when $\gamma = 0$, the graph generated will have $k$ disconnect components: in this case our algorithm trivially recovers the ground-truth clustering, while CLSZ performs poorly.

\subsection{Real-world data sets}
\label{sec:real_exp}
We present experimental results on real-world data sets related to the digital card game Hearthstone, a food web representing species in Florida Bay, and a biological neural network of C. elegans. Numerical comparisons between our algorithm and the CLSZ algorithm on real-world data sets are presented in Table~\ref{tab:RW_comparisons}, showing that our algorithm yields a noticeable improvement.

\begin{table}[!ht]
    \centering
    \caption{Comparison between our algorithm and CLSZ algorithm on a range of real-world data sets when finding $k$ clusters. Results measured using $\delta$ use the matrix representation defined in $M^{\mathcal{S}}_{u,v}$, while results measured using $\delta_P$ use the matrix representation ${M_P^{\mathcal{S}}}_{u,v}$ where intra-cluster edges are penalised}
    \begin{tabular}{ c c c c c c }
    \hline 
    Data set & $k$ & Measure & Our algorithm & CLSZ & \% decrease
    \\
    \hline 
    \multirow{2}{*}{Hearthstone} & \multirow{2}{*}{9} & $\delta$ & 2.030 & 2.253 & 9.9\%
    \\
    & & $\delta_P$ & 1.538 & 1.801 & 14.6\%
    \\
    \hline 
    \multirow{2}{*}{Florida Bay} & \multirow{2}{*}{5} & $\delta$ & 0.066 & 0.142 & 53.5\%
    \\
    & & $\delta_P$ & 0.358 & 0.503 & 28.8\%
    \\
    \hline 
    \multirow{2}{*}{C. elegans} & \multirow{2}{*}{5} & $\delta$ & 0.295 & 0.348 & 15.2\%
    \\
    & & $\delta_P$ & 0.634 & 0.782 & 18.9\%
    \\ \hline
    \end{tabular}
    \label{tab:RW_comparisons}
\end{table}
    

\paragraph{Hearthstone}
In the popular online deck-building card game Hearthstone, two players tactically play cards on their turn with the aim of defeating their opponent. Prior to a match, each player chooses one of 9 possible heroes and designs a deck of 30 cards. Each hero has a unique ``power'', and decks associated with a hero can contain certain cards exclusive to that hero. Different heroes typically require different strategies that work better with their unique cards and power.

A natural condition for a card game to keep remaining engaging over time is to ensure there is no clear best hero or strategy. Moreover, it would be interesting to determine if Hearthstone follows a dynamic similar to ``rock, paper, scissors'', where heroes are strong in some match-ups and weak in others. This dynamic would be reflected in an underlying cyclic meta-graph, demonstrating the relevance of our clustering task to analyse games. 

The data set used in our experiments is from the AAIA'18 Data Mining Challenge \citep{AAIA_2018} which consists of $N=400$ randomly generated decks. Within the data set, $299,680$ matches were played between AIs, who selected their deck uniformly at random at the beginning of each game.

We construct a weighted adjacency matrix $W \in \mathbb{Z}_{\geq 0}^{N \times N}$ to represent this data set as a digraph, where vertices correspond to decks and edges corresponded to the win-rate between them. Let $\widetilde{W}_{u,v}$ be the number of games in which deck $u$ won against deck $v$. The matrix $W$, with $W_{u,v}=(\widetilde{W}_{u,v}-\widetilde{W}_{v,u})/(\widetilde{W}_{u,v}+\widetilde{W}_{v,u})$ then only gives information on the proportion of games won, without information on the number of games played.

\begin{figure}[!ht]
\centering
    \includegraphics[width=0.7\textwidth]{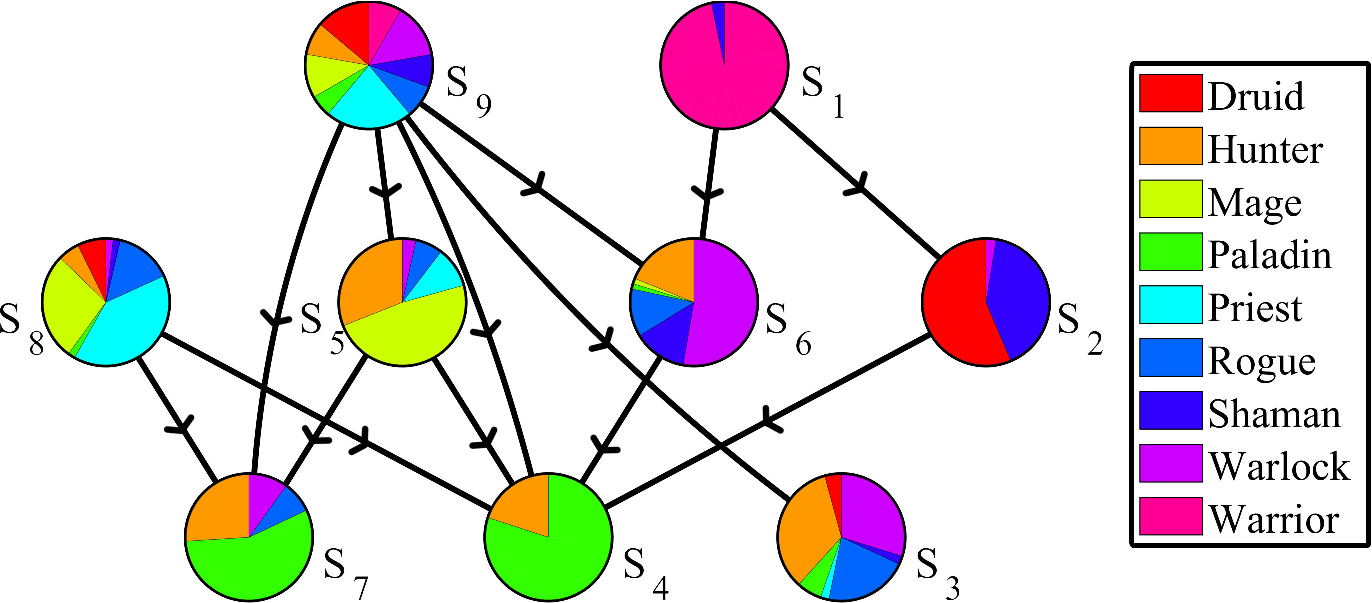}  
    \caption{Hearthstone data set. Clustering partition recovered using our algorithm ($k=9$, $T=100$). Pie charts correspond to clusters  recovered by our algorithm and display the distribution of heroes within. An arrow from cluster $S_i$ to $S_j$ is shown if the edge weight directed from $S_i$ to $S_j$ is more than $80 \%$ of the total edge weight between $S_i$ and $S_j$.}
    \label{Fig:HS_pie}
\end{figure}

Fig.~\ref{Fig:HS_pie} shows the meta-graph computed by our algorithm for $k=9$ (we have chosen a number of clusters equal to the number of heroes). Each pie chart corresponds to a different cluster, with different colours describing the proportion of each hero in the decks belonging to each cluster.  

From the distribution of the heroes in the clusters, we can see that several clusters have one or two dominant heroes, meaning that the clusters found have a non-trivial correlation with the different heroes, i.e., win-rates between  decks are somewhat correlated with their heroes.
In particular, the cluster $S_1$ contains almost all of the Warrior decks.

Furthermore, the meta-graph is acyclic and has a hierarchical structure. The figure suggests that the average Warrior deck does not typically have a disadvantage against any other deck and one could argue that Warrior is the ``best'' hero. On the other hand, an average Paladin deck does not seem to have an advantage against any other deck. However, while the average Paladin deck might perform poorly, the cluster $S_9$, at the top of the hierarchical structure, contains some Paladin decks, as well as decks from all other heroes. This suggests all heroes have potential to succeed if the deck was chosen carefully. This is crucial for the game to be engaging and well-designed. We remark, though, that the games in the data set were played between AIs, which could potentially impose a bias on the results. We also mention that edges connecting clusters at the top of the figure with clusters at the bottom are characterised by a very strong dominant direction: for example, 97.44\% of the edge weight between $S_9$ and $S_4$ is directed from the former to the latter. Conversely, edges between clusters sitting on the same hierarchical level do not have a dominant direction. For example, only 50.43\% of the edge weight between $S_2$ and $S_5$ follows the dominant direction (from $S_5$ to $S_2$).

\begin{figure}[!ht]
\centering
    \includegraphics[width=0.3\textwidth]{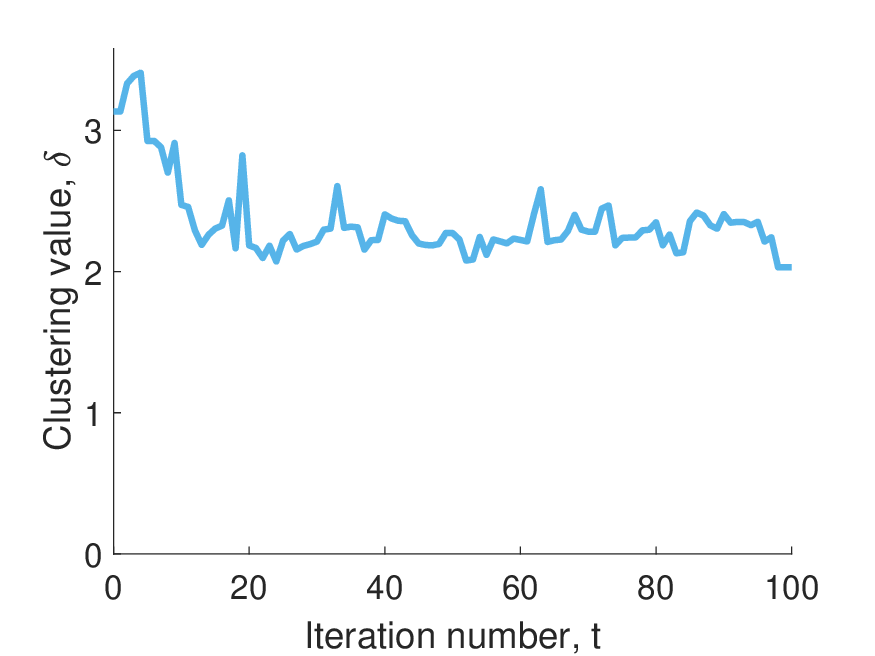}
    \includegraphics[width=0.3\textwidth]{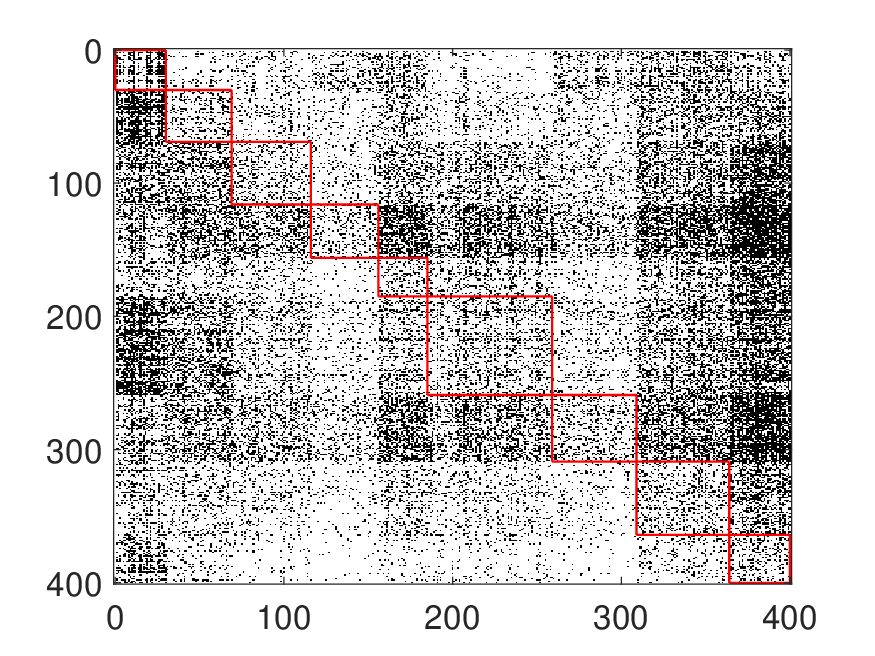}
    \includegraphics[width=0.3\textwidth]{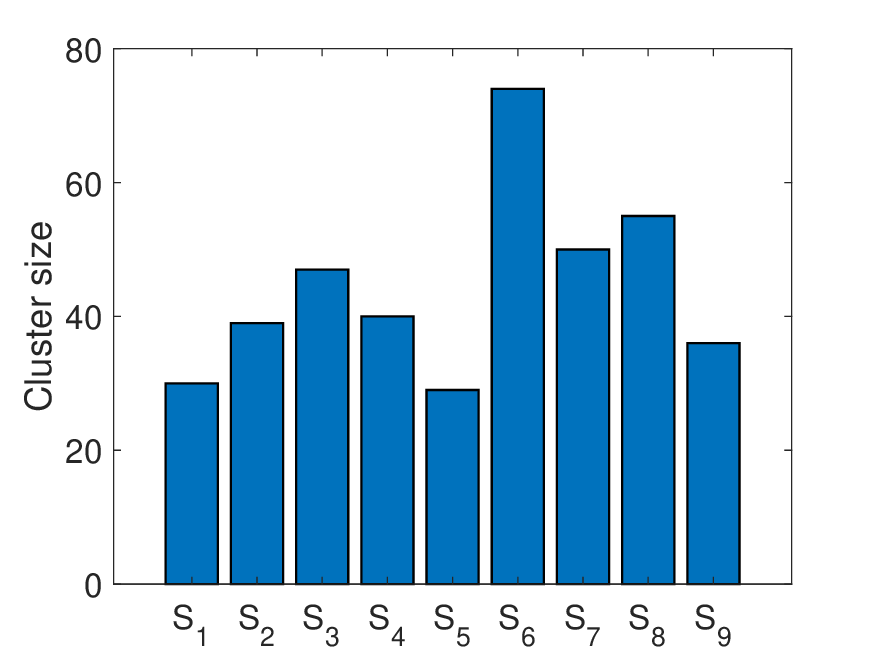}
    \caption{Hearthstone data set, using $k=9$ and $T=100$. Clustering value trajectory of our algorithm (left). Heat map of the adjacency matrix, reordered with respect to the recovered clustering shown in red (middle). Histogram of the sizes of clusters found by our algorithm (right).
    \newline
    Our minimum $\delta$ value: 2.030. CLSZ $\delta$ value: 2.253.}
    \label{Fig:HS}
\end{figure}

Fig.~\ref{Fig:HS} (left) shows that the clustering value of our algorithm improves upon the clustering recovered using the CLSZ algorithm, obtaining a $9.9\%$ reduction in the clustering value. 
Fig.~\ref{Fig:HS} (middle) shows the adjacency matrix after reordering with respect to the clustering partition obtained. It exhibits a clear block model structure that is not block-diagonal. This is in contrast to undirected clustering, where most edges tend to be inside clusters, resulting in a block-diagonal structure. Finally, Fig.~\ref{Fig:HS} (right) displays the histogram of the size of the clusters found by our algorithm.
Interestingly, the clusters all have reasonable size, such that no cluster contains most of the vertices in the graph.


\paragraph{Florida Bay}
Food webs are typically represented as a network with edges pointing in the direction of energy flow, i.e., an arrow from species $u$ to species $v$ suggests $u$ gets consumed by $v$. In this environment, one might typically expect to uncover a directed hierarchical cluster structure, with clusters strongly relating to trophic levels. Species in the same community might have very different diets, but share similar relations with other communities. Ideally, our algorithm will recover this underlying meta-graph.

The Florida Bay data set \citep{snapnets} consists of an unweighted digraph representing carbon exchange in Florida Bay between $N=125$ different species. Edges correspond to a species' typical diet observed within the bay. We run our algorithm for a number of clusters $k=5$. Since species sitting at the same level of the food chain do not typically predate each other, we use the matrix representation ${M_P^{\mathcal{S}}}_{u,v}$ in \eqref{eq:Mdef_pen} that penalises intra-cluster edges.
\begin{figure}[!ht]
\centering
    \includegraphics[width=0.3\textwidth]{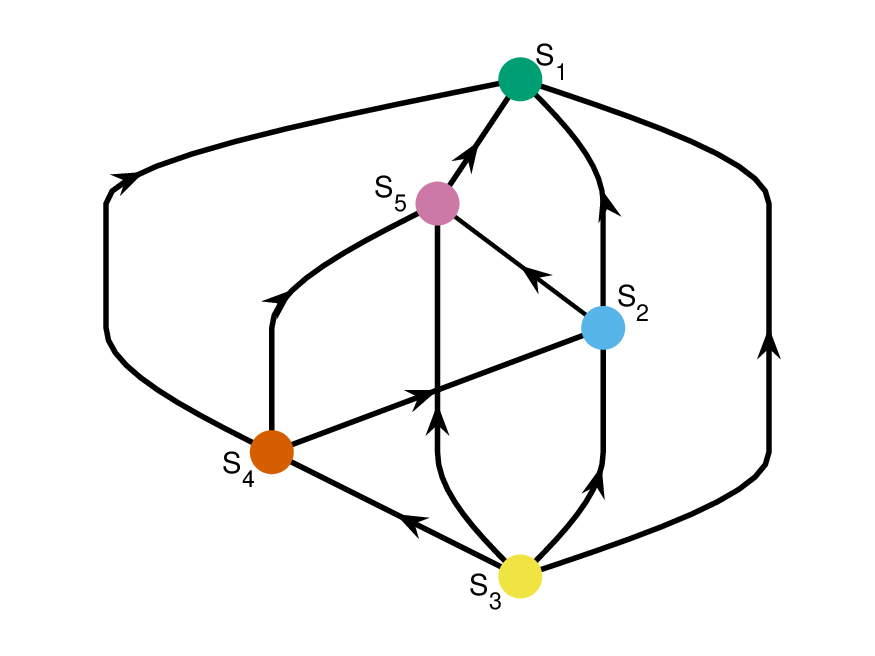}
    \includegraphics[width=0.3\textwidth]{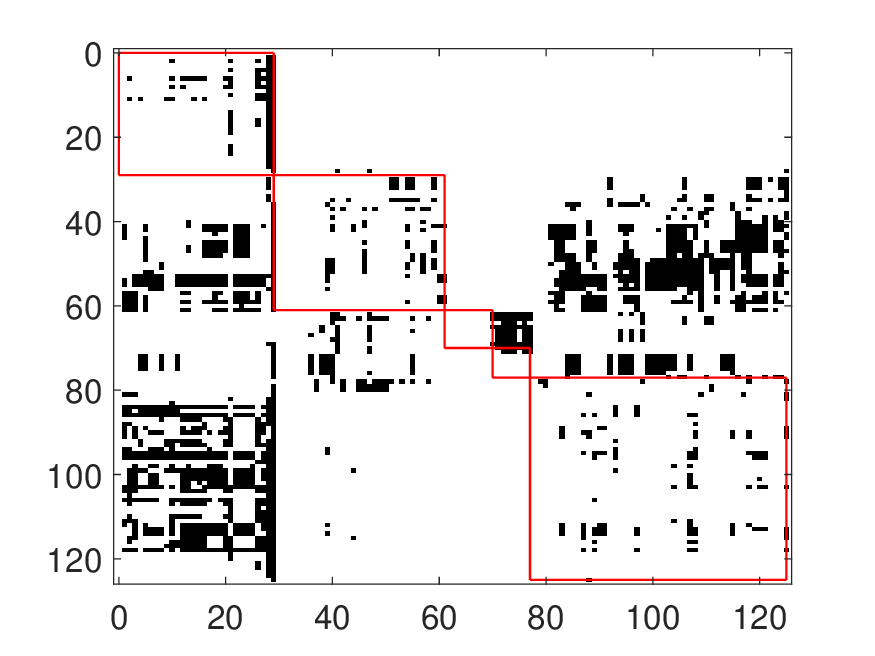}
    \includegraphics[width=0.3\textwidth]{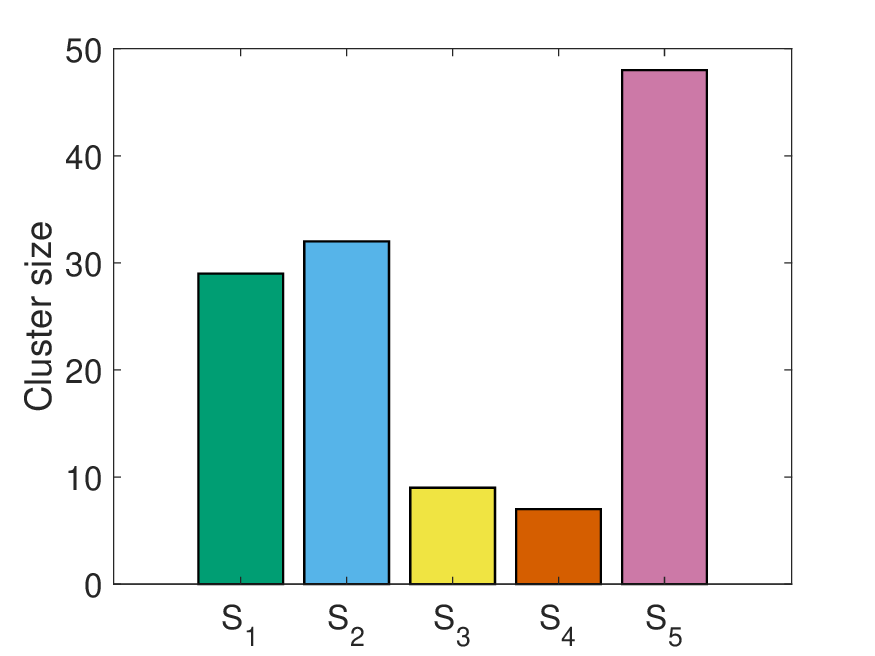}
    \caption{Florida Bay data set using $k=5$ and $T=100$. Meta-graph recovered using our algorithm (left), where an arrow from cluster $S_i$ to $S_j$ implies the edge weight directed from $S_i$ to $S_j$ is more than $80 \%$ of the total edge weight between $S_i$ and $S_j$. Heat map of the adjacency matrix, reordered with respect to the recovered clustering shown in red (middle). Histogram of the sizes of clusters found by our algorithm (right).
    \newline
    Our minimum $\delta_P$ value: 0.358. CLSZ minimum $\delta_P$ value: 0.503.
}
    \label{Fig:Florida}
\end{figure}
Our algorithm considerably improves the clustering value $\delta_P$ with respect to CLSZ ($0.3578$ vs $0.5032$).
Fig.~\ref{Fig:Florida} (left) shows that our algorithm successfully recovered an underlying hierarchical meta-graph. In particular, edges between each pair of clusters are characterised by a very strong dominant direction: for any pair of clusters, at least 94\% percent of the edges between the two clusters are aligned with the corresponding edge in the meta-graph. Seven out of ten pairs of clusters have all cross-edges perfectly following the corresponding edge direction in the meta-graph. The quality of the clustering found by our algorithm is also evident in the heat map of the adjacency matrix, reordered according to our clustering. We can see a clear block structure, characterised by dense regions outside the diagonal blocks.

In addition, we inspected the different species occupying each of the clusters found by our algorithm and we observed the red cluster $S_1$ at the top of the food web in Fig.~\ref{Fig:Florida} (left) contains typical predators, such as crocodiles, birds, and sharks. Furthermore, the cluster $S_3$ at the bottom of the food web contains typical producers such as bacteria and plankton. This is what we would expect a good clustering of the food web to look like. Finally, Fig.~\ref{Fig:Florida} (right) shows that the clusters found by our algorithm also have reasonable sizes, as no cluster contains most of the vertices in the graph.

While perhaps not relevant within the context of a food web, if we do not penalise edges inside the clusters, the best clustering found by our algorithm achieves a clustering value $\delta$  of $0.063$, more than halving the CLSZ clustering value of $0.142$, shown in Table~\ref{tab:RW_comparisons}. 


\paragraph{C. elegans brain network} The C. elegans frontal neural network~\cite{c-elegans} describes the connections between neurons in the brain of C. elegans, a species of worm. It is represented by a  digraph of $130$ vertices and $610$ edges. We have set $k=5$ and applied both versions of our algorithm: the one that penalises edges inside clusters and the one that does not. 
\begin{figure}[!ht]
    \centering
    \includegraphics[width=0.3\textwidth]{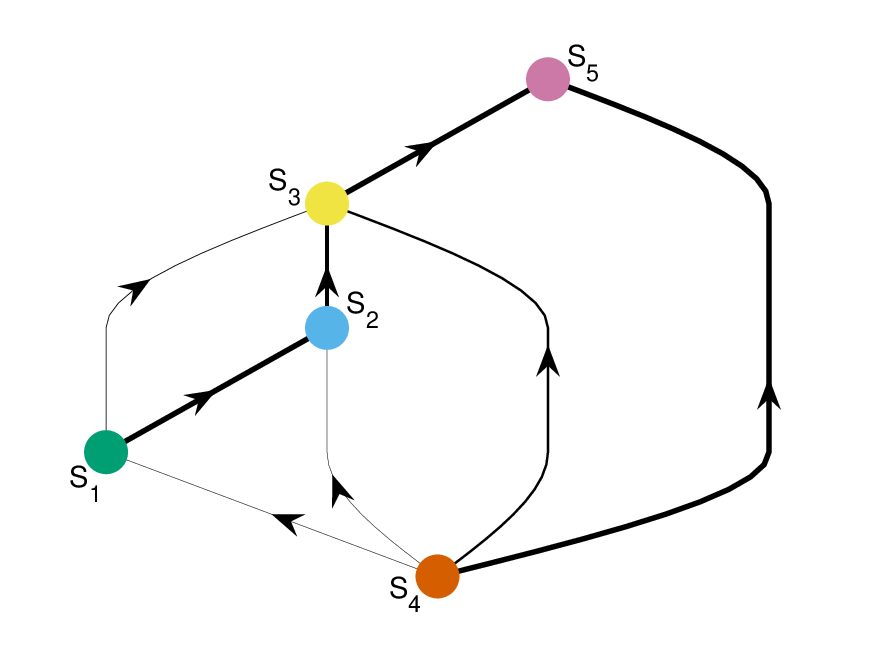}
    \includegraphics[width=0.3\textwidth]{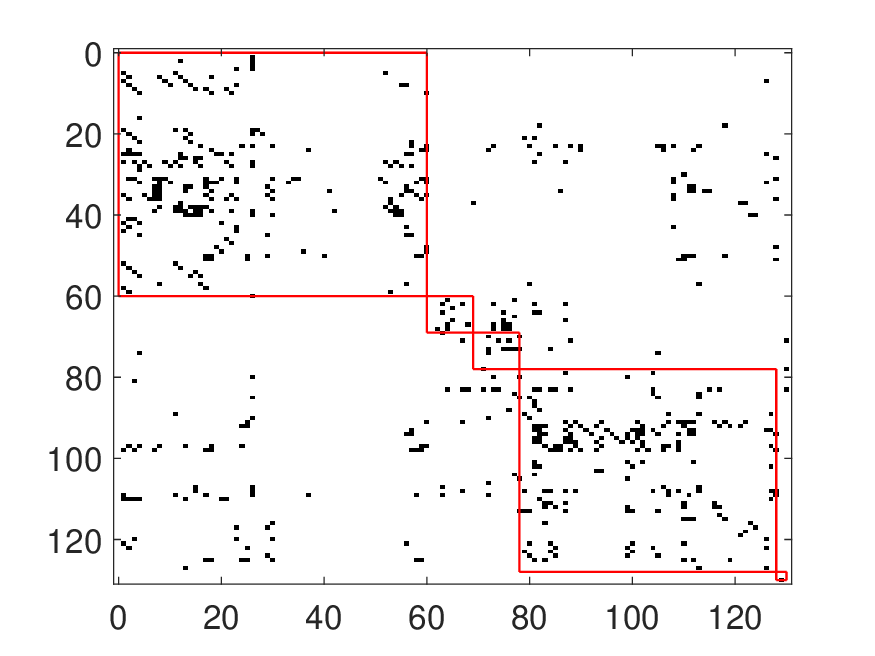}
    \includegraphics[width=0.3\textwidth]{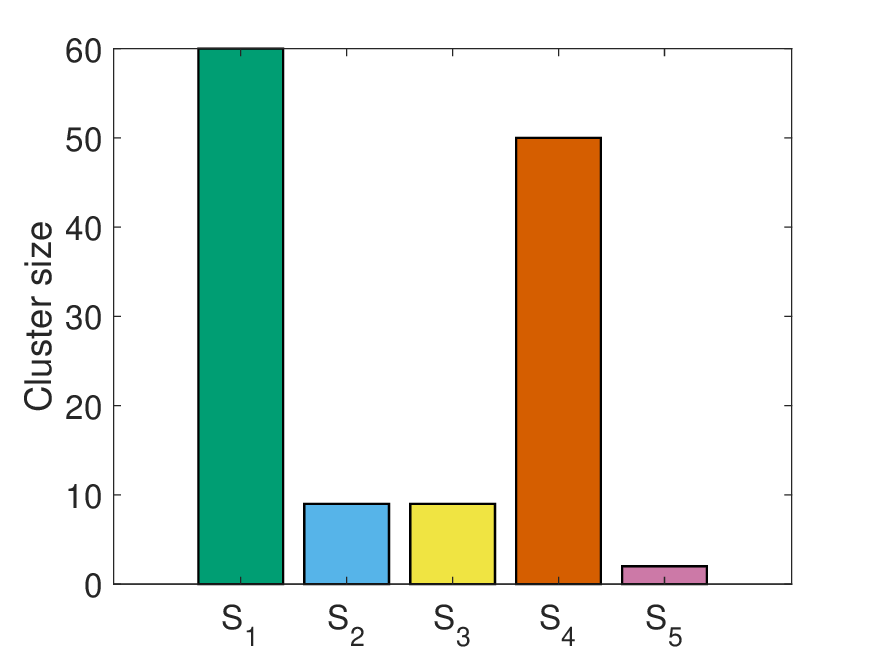}
    \\
    \includegraphics[width=0.3\textwidth]{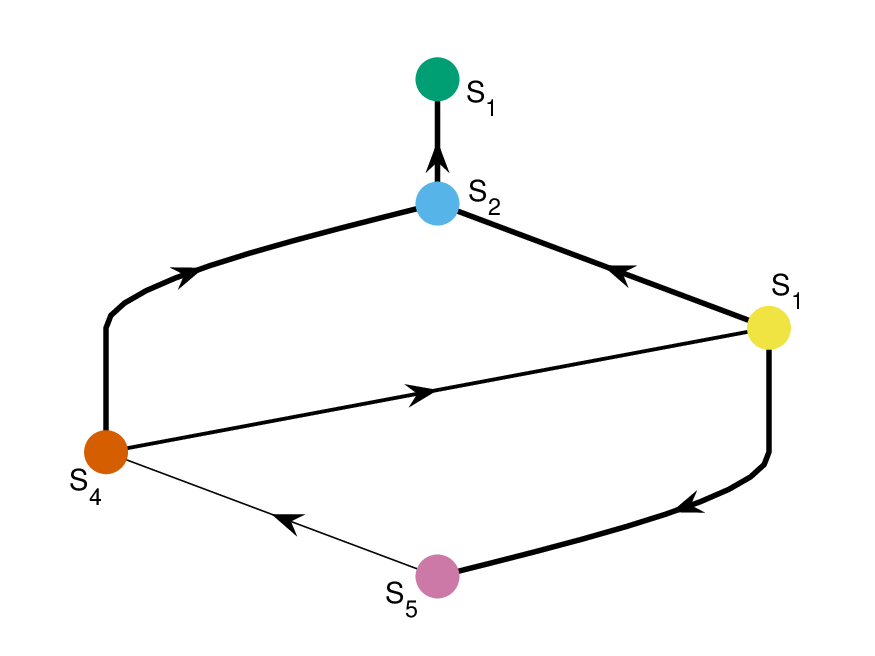}
    \includegraphics[width=0.3\textwidth]{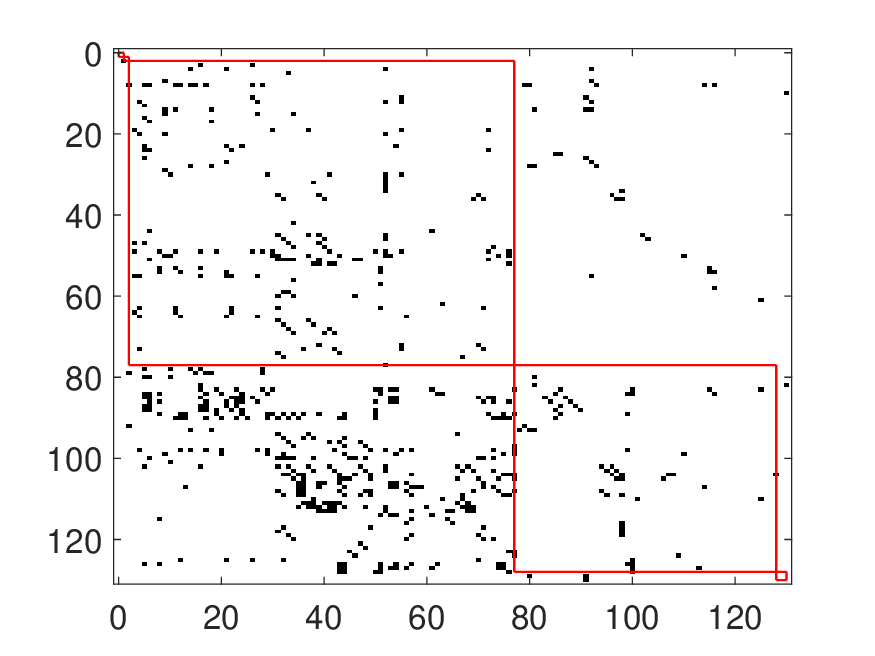}
    \includegraphics[width=0.3\textwidth]{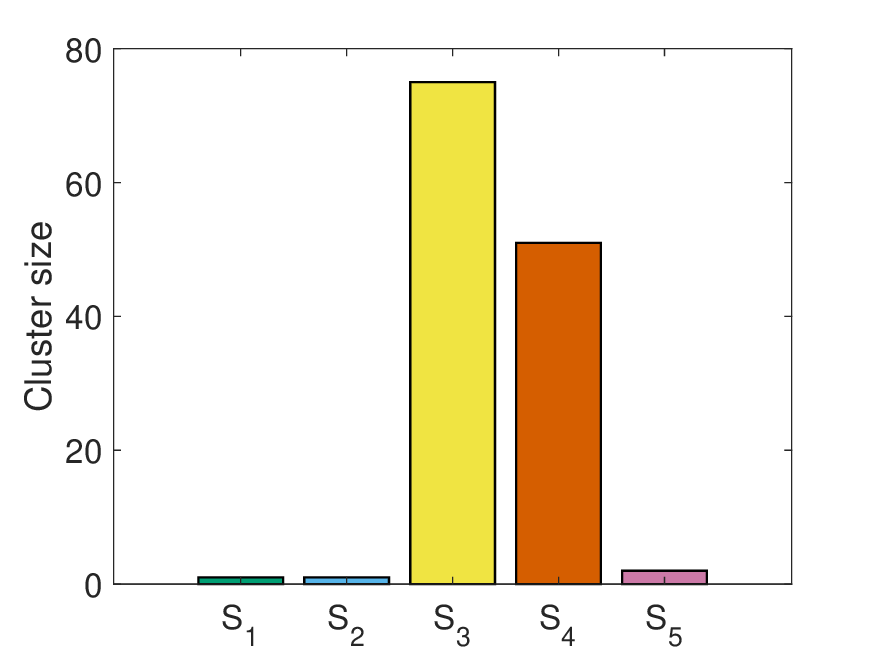}
    \caption{C. elegans data set using $k=5$ and $T=50$. Results from our algorithm without penalties on intra-cluster edges (top) and when penalties on intra-cluster edges are applied (bottom). Meta-graph recovered using our algorithm (left), where an arrow from cluster $S_i$ to $S_j$ implies the edge weight directed from $S_i$ to $S_j$ is more than $50 \%$ of the total edge weight between $S_i$ and $S_j$. Arrow width indicates the proportion of the edge weight oriented in this direction. Heat map of the adjacency matrix, reordered with respect to the recovered clustering shown in red (middle). Histogram of the sizes of clusters found by our algorithm (right).
    \newline
    Our $\delta$ minimum value: 0.295. CLSZ $\delta$ value: 0.348. 
    \newline
    Our $\delta_P$ minimum value: 0.634. CLSZ $\delta_P$ value: 0.782.}
    \label{Fig:celegans}
\end{figure}
The results in Fig.~\ref{Fig:celegans} show that our algorithm achieves a nontrivial improvement over CLSZ on both clustering values $\delta$ and $\delta_p$. The difference in the clusterings found by the two variants of our algorithm is noticeable in the corresponding meta-graphs, as shown in Fig.~\ref{Fig:celegans} (left). The direction of the edges between the two large clusters found without penalising intra-cluster edges is almost random: only 55.65\% of the edges are oriented from $S_4$ to $S_1$. On the contrary, when we penalise intra-cluster edges, 91.16\% of the edges between the two large clusters are oriented from $S_4$ to $S_3$.

This data set can be used 
to showcase the differences in the two variants of our algorithm. As shown by Fig.~\ref{Fig:celegans} (right), in both cases, we find two large clusters and three very small clusters. Fig.~\ref{Fig:celegans} (middle), shows that when we do not penalise edges inside clusters, the reordered adjacency matrix exhibits an almost diagonal block structure. This is consistent with having most of the edges inside the two large clusters. When we do penalise edges inside clusters, instead, the reordered adjacency matrix exhibits a block structure that is far from diagonal, as evidenced by the relatively dense bottom left block. We believe this ability to use different matrix representations to obtain different types of clusters is one of the interesting features of our algorithm.

\section{Conclusion}
We have developed an algorithm for digraph clustering and shown that it outperforms the state-of-the-art on both synthetic and real-world data sets. Our algorithm has the drawback that it requires multiple applications of spectral clustering and is therefore more computationally expensive than algorithms that perform a single spectral clustering application (our algorithm is approximately $T$ times more computationally expensive than CLSZ, where $T$ is the total number of iterations). We believe, however, that this additional computational cost is justified by the improved performances of our algorithm.

We have shown how our algorithm is able to uncover meta-graphs that describe the structure of clusters in a digraph, without imposing any obvious bias on the meta-graphs found. We have also shown how to tweak our algorithm to be able to penalise (or not) clusters with many edges inside. We believe the ability to further tailor our algorithm to specific practical scenarios is one of the most promising directions for future research. In particular, in certain applications we might want to favour specific types of meta-graphs (for example, acyclic meta-graphs). This could be done by, at each iteration, only fitting the respective clustering found to the ``best'' meta-graph among a subset of suitable meta-graphs that satisfy certain conditions specified by the user. This would in turn determine a different matrix representation and orient our algorithm towards finding specific cluster-structure. We remark that this bias would be chosen by the user and not be intrinsic to the algorithm, as instead happens in \cite{cucuringu2020hermitian,LaenenSun}. We leave this direction as future work.

Finally, we mention that finding the right formulation of our clustering task as an optimisation problem (similar to how spectral clustering for undirected graphs can be formulated~\cite{ng2001spectral,ShiMalik}) remains an open question.

\bibliographystyle{comnet}





\end{document}